\documentclass[11pt]{article} 
\usepackage[english]{babel}
\usepackage[utf8]{inputenc}
\usepackage[T1]{fontenc}
\usepackage{xparse}
\usepackage{eucal}
\usepackage{microtype}
\usepackage{csquotes}
\usepackage{amsthm}
\usepackage{mathtools}
\usepackage{amssymb,amsfonts}
\usepackage{tikz}
\usepackage{easy-todo}
\usetikzlibrary{cd}

\usepackage[breaklinks,unicode=true,colorlinks=true]{hyperref}
\hypersetup{urlcolor=blue, citecolor=red, linkcolor=blue}

\usepackage[capitalise]{cleveref} 

\usepackage{geometry}
\usepackage{mathrsfs}

\theoremstyle{plain}
\newtheorem{theorem}{Theorem}[section]
\newtheorem{corollary}[theorem]{Corollary}

\newtheorem{lemma}[theorem]{Lemma}

\theoremstyle{definition}
\newtheorem{definition}[theorem]{Definition}

\DeclarePairedDelimiter{\abs}{|}{|}

\DeclarePairedDelimiter{\set}{\{}{\}}

\newcommand*{\RR}{\mathbb{R}}

\newcommand*{\Cinfty}{\mathscr{C}^\infty}

\renewcommand{\L}{\mathcal{L}}

\DeclareMathOperator{\Id}{Id}

\newcommand{\dd}{\mathrm{d}}

\NewDocumentCommand{\pdv}{mm}{\frac{\partial #1}{\partial #2}}

\NewDocumentCommand{\orth}{om}{{#2}^{\bot\IfValueT{#1}{_#1}}}
\NewDocumentCommand{\orthL}{om}{\prescript{\bot\IfValueT{#1}{_#1}}{}{#2}}

\newcommand*{\lieD}[1]{\mathscr{L}_{#1}}

\DeclareMathOperator{\pr}{pr}

\newcommand{\T}{\mathrm{T}}

\title{The Herglotz variational principle for dissipative field theories}

\author{\sffamily 
Jordi Gaset\thanks{{\bf e}-{\it mail}:
   jordi.gaset@unir.net (ORCID: 0000-0001-8796-3149).} ,
Manuel Lainz\thanks{{\bf e}-{\it mail}:
   manuel.lainz@icmat.es. (ORCID: 0000-0002-2368-5853).} ,
Arnau Mas\thanks{{\bf e}-{\it mail}:
   arnau.mas@physik.uni-muenchen.de (ORCID: 0000-0003-0532-0938).} ,
Xavier Rivas\thanks{{\bf e}-{\it mail}:
   xavier.rivas@unir.net. (ORCID: 0000-0002-4175-5157).} .
}
    \date{\today}

\begin{document}
\maketitle

\begin{abstract}

In the recent years, with the incorporation of contact geometry, there has been a renewed interest in the study of dissipative or non-conservative systems in physics and other areas of applied mathematics. The equations arising when studying contact Hamiltonian systems can also be obtained via the Herglotz variational principle. The contact Lagrangian and Hamiltonian formalisms for mechanical systems has also been generalized to field theories. The main goal of this paper is to develop a generalization of the Herglotz variational principle for first-order and higher-order field theories. In order to illustrate this, we study three examples: the damped vibrating string, the Korteweg--De Vries equation, and an academic example showing that the non-holonomic and the vakonomic variational principles are not fully equivalent.

\end{abstract}

\medskip

\noindent\textbf{Keywords:}
Herglotz variational principle, higher-order field theories, contact field theory, Korteweg--De Vries equation

\noindent\textbf{MSC\,2020 codes:}
37K58, 37L05, 53D10, 35Q53

{
\def\baselinestretch{1}
\small
\def\addvspace#1{\vskip 1pt}
\parskip 0pt plus 0.1mm
\tableofcontents
}

\newpage

\section{Introduction}
It is well known that symplectic geometry is the natural geometric framework to study Hamiltonian mechanical systems \cite{Abr1978,Arnold1997,God1969,Lib1987}. When dealing with time-dependent mechanical systems, cosymplectic geometry is the appropriate framework to work with \cite{Car1993,Chi1994,Ech1991}. These two geometric structures have been generalized to the so-called $k$-symplectic and $k$-cosymplectic structures in order to deal with autonomous and non-autonomous field theories \cite{Awa1992,DeLeo2002,DeLeo1998,DeLeo2001,Mun2010,Rey2005,Rom2007,Rom2011}.

In recent years, the interest in dissipative systems has grown significantly. In part, this is due to the incorporation of contact geometry \cite{Ban2016,Gei2008,Kho2013} to the study of non-conservative Lagrangian and Hamiltonian mechanical systems \cite{Bra2017a,Bra2017b,DeLeo2019b,DeLeo2020b,Gas2019}. This approach has proved to be very useful in many different problems in areas such as thermodynamics, quantum mechanics, general relativity, control theory among others \cite{Bra2018,Cia2018,dLeo2021b,dLeo2017,gaset_EH_2022,Got2016,Kho2013,paiva_generalized_2022,Ram2017,Sim2020,Sus1999}. Recently, the notion of cocontact manifold has been developed in order to introduce explicit dependence on time \cite{DeLeo2022,RiTo-2022}.

This growing interest has driven researchers to look for a generalization of $k$-symplectic and contact geometry in order to work with non-conservative field theories. This new geometric framework is called $k$-contact geometry, and has already been applied to the study of both Hamiltonian and Lagrangian field theories in the autonomous \cite{Gas2020,Gaset2021,PhDThesisXRG} and non-autonomous \cite{Riv2022} cases. The contact formulation of mechanics has also been generalized to describe higher-order mechanical systems in \cite{DeLeo2021}. The Skinner--Rusk formalism has also been studied in detail for both contact \cite{DeLeo2020} and $k$-contact systems \cite{Gra2021}. Recently, the notion of multicontact structure has been introduced \cite{de_leon_multicontact_2022}, generalizing the multisymplectic framework to deal with non-conservative field theories. The Herglotz principle \cite{dLeo2021,Gue1996,Her1930,ryan_when_2022} provides a variational formulation for contact Hamiltonian systems. There have been several attempts \cite{Georgieva2003a,Lazo2018} to generalize this theory to field theories.

In this paper we will derive this principle in a more general geometric language and compare it to the existing approaches. In order to do that, we will review three different formulations of the Herglotz principle for mechanics, the implicit version, the vakonomic version, and the non-holonomic version. In order to find a Herglotz principle for higher dimensions, we will generalize the vakonomic and the non-holonomic versions of the Herglotz principle for mechanical systems. We will see that the non-holonomic approach yields the same field equations as in the $k$-contact \cite{Gaset2021,Riv2022} and multicontact \cite{de_leon_multicontact_2022} formalisms. On the other hand, in contrast to what happens in mechanics, using the vakonomic approach we obtain an additional condition that must be fulfilled. This new equation implies that the $k$-contact and multicontact  Lagrangian formalisms are not fully equivalent to the vakonomic variational principle introduced in the present paper. One of the examples of the last section will illustrate this fact.

The vakonomic Herglotz variational principle for first-order field theories is then extended to a suitable Herglotz principle for higher-order non-conservative field theories. As an example, the Korteweg--De Vries equation \cite{Kor1895} is discussed. This equation arises from a second-order Lagrangian and is used to model waves in shallow waters. In order to have a dissipative behaviour, we add a standard damping term to the Korteweg--De Vries Lagrangian and use the variational principle to derive a non-conservative version of the Korteweg--De Vries equation.

The organization of the paper is as follows.
In first place, Section \ref{sec:2} offers a review of the Herglotz principle in mechanics. In particular, we see three different approaches: the {\it implicit version}, the {\it vakonomic version} and the {\it non-holonomic version}. Section \ref{sec:3} is devoted to extend the Herglotz variational principle from mechanics to field theory using the vakonomic and the non-holonomic approaches. In Section \ref{sec:4} we generalize the results given in Section \ref{sec:3} to the case of higher-order Lagrangian densities using the vakonomic variational principle.

Finally, Section \ref{sec:5} is devoted to study some examples of the theoretical framework developed above. The first example deals with a first-order system consisting of a damped vibrating string with friction linear to the velocity. The second example shows that, as said before, the vakonomic variational principle for field theories and the $k$-contact formulations are not equivalent. We present an academic example consisting in taking the Lagrangian of the previous example and slightly modifying the damping term. In this case, we find a solution to the $k$-contact Euler--Lagrange equations that does not satisfy the additional condition arising from the vakonomic principle. The last example deals with the Korteweg--De Vries equation, which arises from a second-order Lagrangian.

Throughout this paper, all the manifolds are assumed to be real, connected and second countable. Manifolds and mappings are assumed to be smooth. The sum over crossed repeated indices is understood.

\section{The Herglotz principle in mechanics}\label{sec:2}

The Herglotz principle, in simple terms, might be explained as follows. Given a configuration manifold $Q$, consider a Lagrangian function $L:\T Q \times \mathbb{R} \to \mathbb{R}$ depending on the positions $q^i$, the velocities $\dot{q}^i$ and an extra variable $z$ that we can think of as the \emph{action}, but we will soon discuss its meaning in more detail. The Herglotz variational principle states that the trajectory of the system $c(t)$ is a critical point of the action $\zeta(1)$, satisfying $c(0) = q_0$, $c(1) = q_1$, $\zeta(0) = z_0$ and
\begin{equation*}
    \begin{dcases}
        \frac{\dd\zeta}{\dd t} = L(c, \dot{c},\zeta)\,,\\
        \zeta(0) = z_0\,.
    \end{dcases}
\end{equation*}
We note that the action is given by
\begin{equation}
  \zeta(1) = \int_0^1 \frac{\dd \zeta}{\dd t} \dd t + \zeta(0)  = 
  \int_0^1 L(c(t), \dot{c}(t), \zeta(t)) \dd t + z_0\,,
\end{equation}
which, if the Lagrangian does not depend on $z$, coincides with the usual Hamilton's action up to a constant.

A slight modification of this principle, that is the one we will prefer in this paper, is to consider the action as the increment of $z$, that is, 
\begin{equation}
  \zeta(1) - \zeta(0) = \int_0^1 \frac{\dd \zeta}{\dd t} \dd t  = 
  \int_0^1 L(c(t), \dot{c}(t), \zeta(t)) \dd t\,,
\end{equation}
which coincides exactly with Hamilton's action if the Lagrangian $L$ does not depend on $z$. Since both definitions of the action differ only by a constant $z_0$, their critical curves are the same. Indeed, they are the curves $c$ such that $(c, \dot{c}, \zeta)$ satisfy Herglotz's equations:
\begin{equation}
  \frac{\partial L}{\partial q^i} - \frac{\dd }{\dd t} \frac{\partial L}{\partial \dot{q}^i} = \frac{\partial L}{\partial \dot{q}^i} \frac{\partial L}{\partial z}\,.
\end{equation}

To be more precise, we distinguish two possible equivalent interpretations of this principle. We can either understand it as an implicit action principle for curves $c$ on $Q$, or as a constrained but explicit action principle for curves $(c, \zeta)$ on $Q \times \mathbb{R}$. The three different ways to formalize the Herglotz principle that we will see in this section are based on \cite{dLeo2021}. Another version can be found in \cite{de_leon_optimal_2023}.

\subsection{Herglotz principle: implicit version}\label{sect:MecImp}

For the first interpretation, we consider the (infinite dimensional) manifold  $\Omega(q_0, q_1)$  of curves $c:[0,1]\to Q$ with endpoints $q_0,q_1 \in Q$. The tangent space  of $\T_c \Omega(q_0,q_1)$, is the space of vector fields along $c$ vanishing at the endpoints. That is,
\begin{align*}
        \T_c \Omega(q_0,q_1) &=  \{\delta c \mid \delta c(t) \in  \T_{c(t)} Q\,,\ \delta c(0)=0\,,\ \delta c(1)=0 \}\,.
\end{align*}
Let $z_0 \in \mathbb{R}$ and consider the operator
\begin{equation}\label{eq:Z_operator}
    \mathcal{Z}_{z_0}:c\in\Omega(q_0,q_1) \longmapsto \mathcal{Z}_{z_0}(c)\in\Cinfty ([0,1] \to \mathbb{R})\,,
\end{equation}
where $\mathcal{Z}_{z_0}(c)$ is the only solution to the Cauchy problem
\begin{equation}\label{contact_var_ode}
    \begin{dcases}
        \frac{\dd\mathcal{Z}_{z_0}(c)}{\dd t} = L(c, \dot{c}, \mathcal{Z}_{z_0}(c))\,,\\
        \mathcal{Z}_{z_0}(c)(0) = z_0\,,
    \end{dcases}
\end{equation}
that is, it assigns to each curve on the base space its action as a function of time. This map is well-defined because the Cauchy problem~\eqref{contact_var_ode} always has a unique solution.

Now, the \emph{contact action functional} maps each curve $c\in \Omega(q_0, q_1)$ to the increment of the solution of the Cauchy problem \eqref{contact_var_ode}:
\begin{equation}\label{eq:contact_action}
    \begin{aligned}
        \mathcal{A}_{z_0}: \Omega(q_0,q_1) &\longrightarrow \mathbb{R}\\
        c &\longmapsto \mathcal{Z}_{z_0}(c)(1) - \mathcal{Z}_{z_0}(c)(0)\,.
    \end{aligned}
\end{equation}
Note that, by the fundamental theorem of calculus,
\begin{equation*}
    \mathcal{A}_{z_0}(c) = \int_0^1 L(c(t),\dot{c}(t), \mathcal{Z}_{z_0}(c)(t)) \dd t\,.
\end{equation*}
The following theorem states that the critical points of this action functional are precisely the solutions to Herglotz equation~\cite{deLeon2019}.
\begin{theorem}[Herglotz variational principle, implicit version]\label{thm:Herglotz_principle}
    Let $L: \T Q  \times \mathbb{R} \to \mathbb{R}$ be a Lagrangian function and consider $c\in \Omega(q_0, q_1)$ and $z_0 \in \mathbb{R}$. Then, $(c,\dot{c}, \mathcal{Z}_{z_0}(c))$ satisfies the Herglotz equations
    \begin{equation*}  
        \frac{\dd}{\dd t} \frac{\partial L}{\partial \dot{q}^i} - \frac{\partial L}{\partial q^i} = \frac{\partial L}{\partial \dot{q}^i} \frac{\partial L}{\partial z}\,,
    \end{equation*}
    if and only if $c$ is a critical point of the contact action functional $\mathcal{A}_{z_0}$.
\end{theorem}

\begin{proof}
    In order to simplify the notation, let  $\psi = \T_c \mathcal{Z}(\delta v)$. Consider a curve $c_\lambda \in \Omega(q_0,q_1)$, namely a family of curves in $Q$ with fixed endpoints $q_0,q_1$ smoothly parametrized by $\lambda\in\RR$, such that
        \begin{equation*}
            \delta c={\frac{\dd c_\lambda}{\dd \lambda}}\Big{\vert}_{\lambda=0}\,.
        \end{equation*} 
    Since $\mathcal{Z}(c_\lambda)(0)=z_0$ for all $\lambda$, then $\psi(0)=0$.
        
    We compute the derivative of $\psi$ by interchanging the order of the derivatives using the differential equation defining $\mathcal{Z}$:
        \begin{align*}
            \dot{\psi}(t) &= 
            {\frac{\dd }{\dd \lambda}\Big\vert_{\lambda = 0} \frac{\dd }{\dd t} 
            \mathcal{Z}(c_\lambda(t))} =
            \frac{\dd }{\dd \lambda}\Big\vert_{\lambda = 0}{L(c_\lambda(t), \dot{c}_\lambda(t),\mathcal{Z}(c_\lambda)(t))} \\&=
            \frac{\partial L}{\partial q^i}(\chi(t)) {\delta c} ^i(t) +
            \frac{\partial L}{\partial \dot{q}^i}(\chi(t)) {\delta \dot{c}}^i(t) + 
            \frac{\partial L}{\partial z}(\chi(t)) \psi(t)\,.
        \end{align*}
    
    Hence, the function $\psi$ is the solution to the ODE above. Since $\psi(0)=0$, necessarily,
    \begin{equation}
         \psi(t) = \frac{1}{\sigma(t)} \int_0^t \sigma(\tau) \left({
            \frac{\partial L}{\partial q^i}(\chi(\tau)) {\delta c} ^i(\tau) + \frac{\partial L}{\partial \dot{q}^i}(\chi(\tau)) {\delta \dot{c}}^i(\tau)
            } \right) \dd \tau\,,
    \end{equation}
    where
    \begin{equation}
        \sigma(t) = \exp \left({-\int_0^t \frac{\partial L}{\partial z}(\chi(\tau)) \dd \tau}\right) > 0\,.
    \end{equation}
    Integrating by parts and using  and that the variation vanishes at the endpoints, we get the following expression:
    \begin{align*}
            \T_c\mathcal{A}(\delta c) &= \T_c\mathcal{Z}(\delta c)(1) = \psi(1) =  
             \frac{1}{\sigma(1)} \int_0^1  {\delta c}^i (t) \left(
            \sigma(t) \frac{\partial L}{\partial q^i}(\chi(t)) - 
            \frac{\dd}{\dd t} \left(\sigma(t)\frac{\partial L}{\partial \dot{q}^i}(\chi(t))\right) 
            \right) \dd t \\
            &= 
            \int_0^t  \delta c^i(t) \sigma(t) \left( 
             \frac{\partial L}{\partial q^i}(\chi(t)) + 
            \frac{\dd}{\dd t} \frac{\partial L}{\partial \dot{q}^i}(\chi(t)) -
             \frac{\partial L}{\partial \dot{q}^i}(\chi(t))  \frac{\partial L}{\partial z}(\chi(t))
            \right) \dd t\,,
    \end{align*}
    where we have used that
    \begin{equation}
        \frac{\dd \sigma}{\dd t}(t) = -  \frac{\partial L}{\partial z}(\chi(t)) \sigma(t)\,.
    \end{equation}
    Since this must hold for every possible variation, we have
    \begin{equation*}
         \sigma(t) \left( 
             \frac{\partial L}{\partial q^i}(\chi(t)) + 
            \frac{\dd}{\dd t} \frac{\partial L}{\partial \dot{q}^i}(\chi(t)) -
             \frac{\partial L}{\partial \dot{q}^i}(\chi(t))  \frac{\partial L}{\partial z}(\chi(t)) \right) = 0\,,
    \end{equation*}
    thus obtaining the Herglotz equation.
\end{proof}
    
\subsection{Herglotz principle: vakonomic version}\label{subs:mechanicsLagMult}

Another way to understand this principle is to think of it as a constrained variational principle for curves on $Q \times \mathbb{R}$. This time, we will work on the manifold $\widetilde\Omega(q_0,q_1,z_0)$  of curves $\widetilde c = (c, \zeta):[0,1]\to Q \times \mathbb{R}$ such that $c(0)=q_0$, $c(1)=q_1$, $\zeta(0)=z_0$. Note that we do not constraint $\zeta(1)$. The tangent space at the curve $\widetilde c\in\widetilde\Omega(q_0,q_1,z_0)$ is given by
\begin{align}
    \T_{\widetilde c} \widetilde\Omega(q_0,q_1,z_0) &=  \{
        \delta \widetilde c(t) = (\delta c(t), \delta \zeta(t)) \in  \T_{\widetilde c(t)} (Q \times \mathbb{R})  \mid \delta c(0)=0, \, \delta c(1)=0, \delta \zeta(0) = 0  \}\,.
\end{align}
In this space, the action functional $\widetilde{\mathcal{A}}$ can be defined as an integral
\begin{equation}
    \begin{aligned}
        \widetilde{\mathcal{A}}: \widetilde\Omega(q_0, q_1,z_0) &\longrightarrow \mathbb{R}\\
        \widetilde{c} & \longmapsto \zeta(1) - \zeta(0) = \int_0^1 \dot{\zeta}(t)  \dd t\,.
    \end{aligned}
\end{equation}

We will restrict this action to the set of paths that satisfy the constraint $\dot{\zeta} = L$. For this, we consider the paths at the zero set of the constraint function $\phi_L$:
\begin{equation}\label{eq:lag_constraints}
    \phi_L(q,\dot{q},z,\dot{z}) = \dot{z} - L(q,\dot{q},z)\,.
\end{equation}
That is, we consider
\begin{equation}
    \widetilde\Omega_L (q_0, q_1,z_0)= \{ \widetilde{c} = (c, \zeta) \in \widetilde\Omega(q_0, q_1,z_0) \mid \phi_L \circ \dot{\widetilde{c}} = \dot{\zeta} - L(c,\dot{c}, \zeta) = 0 \}\,.
\end{equation}

Note that, since the Cauchy problem~\eqref{contact_var_ode} has a unique solution, the elements $(c, \zeta) \in \widetilde\Omega_L(q_0, q_1,z_0)$ are precisely $(c, \mathcal{Z}_{z_0}(c))$, where $c \in \Omega(q_0,q_1)$. That is, the map $\Id \times \mathcal{Z}_{z_0}: \Omega(q_0,q_1) \to \widetilde\Omega_L(q_0,q_1,z_0)$ given by ${(\Id \times \mathcal{Z}_{z_0})}(c) = (c, \mathcal{Z}_{z_0}(c))$ is a bijection, with inverse $(\pr_Q)_* (c,\zeta) = c$. Moreover, the following diagram commutes
\begin{equation}
    \begin{tikzcd}
        & \mathbb{R} &                                                             \\
{\Omega(q_0, q_1)} \arrow[rr, "\Id \times \mathcal{Z}_{z_0}"] \arrow[ru, "\mathcal{A}"] &            & {\widetilde\Omega_L (q_0, q_1,z_0)} \arrow[lu, "\widetilde{\mathcal{A}}"']
\end{tikzcd}\end{equation}
Hence $({c}, \zeta) \in \widetilde\Omega_L (q_0, q_1,z_0)$ is a critical point of the functional $\widetilde{\mathcal{A}}$ if and only if $c$ is a critical point of $\mathcal{A}$. So the critical points of $\mathcal{A}$ restricted to $\widetilde\Omega_L (q_0,q_1,z_0)$ are precisely the curves that satisfy the Herglotz equations.

\begin{theorem}[Herglotz variational principle, vakonomic version]\label{thm:Herglotz_principle_constrained}
    Let $L: \T Q  \times \mathbb{R} \to \mathbb{R}$ be a Lagrangian function and let $(c,\zeta)\in \widetilde\Omega_L(q_0, q_1, z_0)$. Then, $(c,\dot{c}, \zeta)$ satisfies the Herglotz equations:
    \begin{equation}\label{eq:Herglotz-eq-mechanics-constrained} 
            \frac{\dd}{\dd t} \frac{\partial L}{\partial \dot{q}^i} 
            - \frac{\partial L}{\partial q^i}=
            \frac{\partial L}{\partial \dot{q}^i} \frac{\partial L}{\partial z}\,,
    \end{equation}
    if and only if $(c,\zeta)$ is a critical point of $\widetilde{\mathcal{A}}\vert_{\widetilde\Omega_L(q_0, q_1, z_0)}$.
\end{theorem}

We will provide an alternative proof by directly finding the critical points of the functional $\widetilde{\mathcal{A}}$ restricted to $\widetilde\Omega_L (q_0, q_1,z_0) \subseteq \widetilde\Omega (q_0, q_1,z_0))$ using the following infinite-dimensional version of the Lagrange multiplier theorem (see \cite{Arnold1997} for more details).
\begin{theorem}[Lagrange multiplier Theorem]\label{thm:Lagrange_multipliers}
    Let $M$ be a smooth manifold and let $E$ be a Banach space. Consider a smooth submersion $g:M \to E$ such that $A=g^{-1}(\{0\})$ is a smooth submanifold, and a smooth function $f:M \to \mathbb{R}$. Then $p \in A$ is a critical point of $f \vert_A$ if and only if there exists $\widehat{\lambda} \in E^*$ such that $p$ is a critical point of $f + \widehat{\lambda} \circ g$. 
\end{theorem}

\begin{proof}[Proof of Herglotz variational principle, vakonomic version]

    We will apply this result to our situation. In the notation of the theorem, $M= \widetilde{\Omega}(q_0,q_1,z_0)$ is the smooth manifold. We pick the Banach space $E = L^2([0,1] \to \mathbb{R})$ of square integrable functions. This space is, indeed, a Hilbert space with inner product
    \begin{equation*}
        \langle \alpha, \beta \rangle = \int_0^1 \alpha(t) \beta (t) \dd t\,.
    \end{equation*}
    Recall that, by the Riesz representation theorem, there exists a bijection between $L^2([0,1] \to \mathbb{R})$ and its dual such that for every $\widehat{\alpha} \in L^2([0,1] \to \mathbb{R})^*$ there exists $\alpha \in L^2([0,1] \to \mathbb{R})$ such that $\hat{\alpha}(\beta)=  \langle \alpha, \beta \rangle$ for all $\beta \in L^2([0,1] \to \mathbb{R})$.
    
    Our constraint function is
    \begin{equation*}
        \begin{aligned}
            g: \widetilde{\Omega}(q_0,q_1,z_0) &\longrightarrow  L^2([0,1] \to \mathbb{R})\\
            \widetilde{c} &\longmapsto
            (\phi_L) \circ (\widetilde{c}, \dot{\widetilde{c}})\,,
        \end{aligned}
    \end{equation*}
    where $\phi_L$ is a constraint function locally defining $A = g^{-1}(0) = \widetilde{\Omega}_L(q_0,q_1,z_0)$.
    
    By Theorem~\ref{thm:Lagrange_multipliers}, $c$ is a critical point of $f=\widetilde{\mathcal{A}}$ restricted to $\widetilde{\Omega}_L (q_0,q_1,z_0)$ if and only if there exists $\widehat{\lambda} \in L^2([0,1] \to \mathbb{R})^*$ (which is represented by $\lambda \in L^2([0,1] \to \mathbb{R})$) such that $c$ is a critical point of $\widetilde{\mathcal{A}}_\lambda = \widetilde{\mathcal{A}} + \widehat{\lambda} \circ g$.
    
    Indeed,
    \begin{equation*}
        \widetilde{\mathcal{A}}_\lambda = \int_0^1 L_\lambda(\widetilde{c}(t), \dot{\widetilde{c}}(t)) \dd t\,,
    \end{equation*}
    where
    \begin{equation*}
        L_\lambda(q,z,\dot{q},\dot{z})= \dot{z} - \lambda \phi_L(q,z, \dot{q},\dot{z})\,.
    \end{equation*}
    
    Since the endpoint of $\zeta$ is not fixed, the critical points of this functional $\widetilde{\mathcal{A}}_\lambda$ are the solutions of the Euler--Lagrange equations for $L_\lambda$ that satisfy the natural boundary condition
    \begin{equation*}
        \frac{\partial L_\lambda}{\partial \dot{z}}(\widetilde{c}(1),\dot{\widetilde{c}}(1)) = 1- \lambda(1) \frac{\partial \phi_L}{\partial \dot{z}}(\widetilde{c}(1),\dot{\widetilde{c}}(1)) = 0\,.
    \end{equation*}
    Since $\phi_L = \dot{z}-L$, this condition reduces to $\lambda(1)=1$.
    
    The Euler--Lagrange equations of $L$ are given by
    \begin{subequations}\label{eq:euler_langrange_contact}
        \begin{align} \label{eq:euler_langrange_q}
            \frac{\dd }{\dd t} \left(\lambda(t)  
            \frac{\partial \phi_L(\widetilde{c}(t),\dot{\widetilde{c}}(t))}{\partial \dot{q}^i}  \right) - \lambda(t) \frac{\partial \phi_L(\widetilde{c}(t),\dot{\widetilde{c}}(t))}{\partial q^i} &= 0\,, \\
            \frac{\dd }{\dd t} \left(\lambda(t)     \label{eq:euler_langrange_z}
            \frac{\partial \phi_L(\widetilde{c}(t),\dot{\widetilde{c}}(t))}{\partial \dot{z}}  \right) - \lambda(t) \frac{\partial \phi_L(\widetilde{c}(t),\dot{\widetilde{c}}(t))}{\partial z} &= 0\,,
        \end{align}
    \end{subequations}
    Since $\phi_L = \dot{z} - L$, the equation \eqref{eq:euler_langrange_z} for $z$ is just
    \begin{equation*}
        \frac{\dd \lambda(t)}{\dd t} = - \lambda(t) \frac{\partial L}{\partial z}\,.
    \end{equation*}
    Substituting on \eqref{eq:euler_langrange_q} and dividing by $\lambda$, we obtain the Herglotz equations \eqref{eq:Herglotz-eq-mechanics-constrained}.
\end{proof}

\subsection{Herglotz principle: nonholonomic version}
     Another way to obtain the Herglotz equation of motion is through a non-linear non-holonomic principle, the so-called Chetaev principle \cite{Gra2003}. Instead of restricting the space of admissible curves $\widetilde{\Omega}_L(q_0,q_1,z_0) \subseteq \Omega(q_0,q_1,z_0)$ and find the critical points on this submanifold, we directly restrict the space of admissible variations, so that the differential of the action has to vanish only in a selection of variations. Hence, the solutions of this principle are not necessarily critical points of the action functional restricted to any space.

\begin{definition}\label{dfn:noholonomicmech}
    A section $\widetilde c = (c, c_z) \in \widetilde{\Omega}_L(q_0,q_1, z_0)$ satisfies the \emph{non-holonomic Herglotz variational principle} if $\T_{\widetilde{c}}\mathcal A(\delta c)=0\,$ for all vector fields $\delta c\in \T_c\widetilde{\Omega}(q_0,q_1,z_0)$ such that $\dd \phi_L(\mathcal{I}(\delta c))=0$, where $\mathcal{I}$ denotes the vertical endomorphism of $\T(\T(Q\times\mathbb{R}))$.
\end{definition}

 If $\delta \widetilde{c}=\delta q^i\dfrac{\partial}{\partial q^i} + \delta z\dfrac{\partial}{\partial z}$, then
$$
\dd \phi_L(\mathcal{I}(\delta c)) = \dd \phi_L\left(\delta q^i\frac{\partial}{\partial \dot{q}^i} + \delta z\frac{\partial}{\partial \dot{z}}\right)=\delta z-\delta q^i\frac{\partial L}{\partial \dot{q}^i}\,.
$$
Then, the nonholonomic dynamics are given by \cite{cendra2004,Gra2003}.

\begin{theorem}[Herglotz's variational principle, nonholonomic version]\label{thm:herglotz-mec-nonhol}
    Let $L: \T Q \times \RR \to \RR$ be a Lagrangian function and let $\widetilde c = (c, c_z) \in \widetilde{\Omega}_L(q_0,q_1, z_0)$. Then $\widetilde c$ satisfy the non-holonomic Herglotz variational principle if, and only if, $(c,\dot{c},c_z)$ satisfies Herglotz's equations:
    \begin{equation}
        \begin{aligned}
              \frac{\dd }{\dd t} \left({\pdv{L}{\dot{q}^i}}  \right) - \pdv{L}{q^i}  &=
              \pdv{L}{\dot{q}^i} \pdv{L}{z}\,,\\
              \dot{z} &= L\,.
        \end{aligned}
    \end{equation}
\end{theorem}

\section{The Herglotz principle for fields}\label{sec:3}

In the literature, there exists a non-covariant formulation of the Herglotz principle for fields theories~\cite{Georgieva2003a}. A more general approach is given in \cite{Lazo2018}, although only a class of Lagrangian functions is considered, which we call Lagrangians with closed action dependence (see Definition \ref{dfn:ClosedLag}).

The method presented in \cite{Lazo2018} uses an implicit argument, similar to the method presented in Section \ref{sect:MecImp} for mechanical systems. We propose two alternative methods: the non-holonomic principle, which is compatible with the $k$-contact \cite{Gas2020,Gaset2021}, $k$-cocontact \cite{Riv2022} and multicontact \cite{de_leon_multicontact_2022} formulations; and the vakonomic principle, which can be extended to higher-order Lagrangians.

Consider a Lagrangian function $L(x^\mu,u^a,u^a_\mu,z^\mu)$ depending on the coordinates $(x^\mu)$ of an $m$-dimensional spacetime $M$, the values of fields $u^a$, their derivatives $u^a_\mu$ at the point $x$ and the variables $z^\mu$ that, in this context do not represent the action, but the \emph{action density}. In order to compute the action of a local field $\sigma$ defined on $D \subseteq M$, we find a vector field $\zeta^\mu$ such that
\begin{equation}\label{eq:zmu_def0}
    D_\mu \zeta^\mu = L\,.
\end{equation}
Then, the action is
\begin{equation}\label{eq:herglotz_field_0}
    \int L \dd^n x = \int D_\mu \zeta^\mu  \dd^n x
    = \int_{\partial D} \zeta^\mu  \eta_\mu \dd \sigma\,,
\end{equation}
where $\eta_\mu$ is the normal unit vector to the surface and $\dd \sigma$ is the surface differential. The last equality follows from Stokes' Theorem. Note that if $M$ is one-dimensional, the action is just $\zeta(1)-\zeta(0)$, and thus we recover the Herglotz action for mechanical systems.

The critical points of this action along the local fields $\sigma$ with the same values on the boundary would be the solutions to the Herglotz field equations
\begin{equation}\label{eq:Herglotz-equations}
    D_\mu \left( \frac{\partial L}{\partial u^a_\mu} \right) - \frac{\partial L}{\partial u^a}  = \frac{\partial L}{\partial u^a}  \frac{\partial L}{\partial z^\mu}\,.
\end{equation}

These equations are obtained in \cite{Lazo2018} through an implicit argument, in a similar spirit to the proof of~\cref{thm:Herglotz_principle}. Note that the Lagrangian theory of $k$-contact fields~\cite{Gaset2021} provides the same equations.

However, we find two issues on this derivation of the variational principle. First of all, the definition of $z^\mu$ in equation \eqref{eq:zmu_def0} depends on a metric on $M$ in order to compute its divergence. This can be easily fixed by taking $z^\mu$ to be components of a $(k-1)$-differential form instead of a vector field.

The second issue is more subtle. The solution of~\eqref{eq:zmu_def0} is not unique, and hence the action is not well-defined. This is not a problem if the Lagrangian does not depend on $\zeta^\mu$, because in this case all the solutions to~\eqref{eq:zmu_def0} differ only by an exact term, whose integral is zero, and does not contribute to the action, but this is not true in general. Indeed, $\zeta$ may appear in equation~\eqref{eq:herglotz_field_0}. In~\cite{Lazo2018} the authors assume some conditions on the Lagrangian in order to find a unique solution. Moreover, \eqref{eq:zmu_def0} might have no solutions and hence we will need to add more constraints in order to ensure the existence of solutions.

One way to fix this problem is to prescribe boundary conditions on~\eqref{eq:zmu_def0} that make the solution unique. However, we will avoid this problem choosing a \enquote{constrained formulation} of this problem, in the same spirit of Theorems \ref{thm:Herglotz_principle_constrained} and \ref{thm:herglotz-mec-nonhol}, instead of the \enquote{implicit} approach used in~\cite{Lazo2018}.
\subsection{Geometric structures}

Let $M$ be and $m$-dimensional orientable manifold representing the spacetime and consider a fiber bundle $E \to M$. Let $(x^\mu, u^a)$ be adapted coordinates on $E$ and let $\dd^mx=\dd x^1\wedge\dots\wedge\dd x^m$ be a volume form on $M$. Then, we will denote $\dd^{m-1}x_\mu=i_{\frac{\partial}{\partial x^\mu}}\dd^mx\in\Omega^{m-1}(M)$. 
The configuration space is the bundle $\pi:E\times_M\Lambda^{m-1}M\rightarrow M$, because the action densities are $(m-1)$-forms on $M$. The adapted coordinates of the first jet bundle $J^1(E\times_M\Lambda^{m-1}M)$ are $(x^\mu,u^a,u^a_\mu,z^\nu,z^\nu_\mu)$, where $z^\nu$ are the coordinates of $\Lambda^{m-1}(M)$ induced by the local basis $\left\{\dd^{m-1}x_\nu\right\}_{\nu=1,\dots,m}$. We consider the first jet of the action densities because it is necessary to intrinsically define the constraint~\eqref{eq:zmu_def0}.

Given a coordinate system, the total derivative $D_\mu:\Cinfty(J^1(E \times_M \Lambda^{m-1}M))\rightarrow \Cinfty(J^{2}(E \times_M \Lambda^{m-1}M))$, for $\mu=1,\dots,m$, is a derivation given by
\begin{equation*}
D_\mu f=\frac{\partial f}{\partial x^\mu}+u^a_{\mu}\frac{\partial f}{\partial u^a}+z^\nu_{\mu}\frac{\partial f}{\partial z^\nu}+u^a_{\tau\mu}\frac{\partial f}{\partial u^a_\tau}+z^\nu_{\tau\mu}\frac{\partial f}{\partial z^\nu_\tau}\,,
\end{equation*}
where $f\in \Cinfty(J^1(E \times_M \Lambda^{m-1}M))$.

For any section $\rho:M\rightarrow E\times_M\Lambda^{m-1}M$, the total derivative satisfies the property
$$(j^2\rho)^*(D_\mu f)=\frac{\partial (j^1\rho)^*f}{\partial x^\mu}\,.
$$

Given a vector field $\xi\in\mathfrak{X}(E\times_M\Lambda^{m-1}M)$ with local flow $\gamma_r:E\times_M\Lambda^{m-1}M\rightarrow E\times_M\Lambda^{m-1}M$, its complete lift to $J^1(E\times_M\Lambda^{m-1}M)$  is the vector field $\xi^1 \in\mathfrak{X}(J^1(E\times_M\Lambda^{m-1}M))$ whose local flow is $j^1\gamma_r$. If $\xi\in\mathfrak{X}(E\times_M\Lambda^{m-1}M)$ is a vertical vector field with respect to the projection $\pi$ with local expression
$$
    \xi=\xi^a\frac{\partial}{\partial u^a}+\xi^\nu\frac{\partial}{\partial z^\nu}\,,
$$
its complete lift is
$$
    \xi^1=\xi^a\frac{\partial}{\partial u^a}+\left(\frac{\partial\xi^a}{\partial x^\mu}+u^b_\mu\frac{\partial\xi^a}{\partial u^b}+z^\tau_\mu\frac{\partial\xi^a}{\partial z^\tau}\right)\frac{\partial}{\partial u^a_\mu}+\xi^\nu\frac{\partial}{\partial z^\nu}+\left(\frac{\partial\xi^\nu}{\partial x^\mu}+u^b_\mu\frac{\partial\xi^\nu}{\partial u^b}+z^\tau_\mu\frac{\partial\xi^\nu}{\partial z^\tau}\right)\frac{\partial}{\partial z^\nu_\mu}\,.
$$

The Lagrangian density $\mathcal{L}: J^1(E\times_M\Lambda^{m-1}M) \to \Lambda^m M$ is a fiber bundle morphism over $M$. In local coordinates, $\mathcal{L}(x^\mu, u^a, u^a_\mu, z^\mu) = L(x^\mu, u^a, u^a_\mu, z^\mu) \dd^m x$.
In order to define intrinsically the constraint~\eqref{eq:zmu_def0}, we define the \emph{canonical differential action form} as
\begin{align*}
    \overline{DS}:J^1\Lambda^{m-1}M & \rightarrow\Lambda^m(M)
    \\
    j^1\alpha & \longmapsto \dd\alpha\,.
\end{align*}
The name is inspired by the canonical action form introduced in \cite{de_leon_multicontact_2022}. In local coordinates, it reads
$$
    \overline{DS}(z^\nu,z^\nu_\mu)=z^\mu_\mu\dd^mx\,.
$$
Then, the constraint ~\eqref{eq:zmu_def0} can be written as
\begin{equation}
    \Phi=\tau^*\overline{DS}-\mathcal{L}=0\,,
\end{equation}
where $\tau:J^1(E\times_M\Lambda^{m-1}M)\rightarrow J^1\Lambda^{m-1}M$ is the natural projection. In local coordinates, $\Phi=\phi\dd^mx$, with $\phi=z^\mu_\mu-L$. The situation is described by the following commutative diagram
\begin{equation*}
\begin{tikzcd}
    & J^1(E\times_M\Lambda^{m-1}M) \arrow[ld] \arrow[d, "\pi^1"] \arrow[rd, "\tau"'] \arrow[rrd, "\mathcal{L}"] \arrow[ddd, "\bar\pi"', bend right=55, shift right=2] &            &        \\
J^1E \arrow[d] & E\times_M\Lambda^{m-1}M \arrow[dd, "\pi"]    & J^1\Lambda^{m-1}M \arrow[d] \arrow[r, "\overline{DS}"'] & \Lambda^mM \\
E \arrow[rd]   &     & \Lambda^{m-1}M \arrow[ld]   &     \\
 & M \arrow[uu, "\rho"', bend right] \arrow[uuu, "j^1\rho"', bend right=55, shift right=2] \arrow[lu, "\sigma", bend left] \arrow[ru, "\zeta"', bend right]       &     &      
 \end{tikzcd}
\end{equation*}

Given a submanifold $D\subset M$, the set of sections that satisfy the constraint $\Phi$ is denoted by
\begin{equation*}
 \Omega= \set{\rho \in \Gamma_D(E \times_M \Lambda^{m-1} M)\, \text{ such that }\, (j^1\rho)^*\Phi=0}\,.
\end{equation*}
Then, the action associated to $\mathcal{L}$ is:
\begin{align*}
    \mathcal{A}:\Omega&\longrightarrow \mathbb{R} \\
    \rho & \longmapsto\int_D (j^1\rho)^*\mathcal{L}\,.
\end{align*}
In general, the variations of this action are the elements tangent to $\rho$ which vanish at $\partial D$, which can be seen as the $\pi$-vertical vector fields along $\rho$. Thus, we define:
\begin{align*}
  \T_\rho\Gamma_D = \{\xi:D\rightarrow \T(E \times_M \Lambda^{m-1} M)\mid \xi(x)\in \T_{\rho(x)}(E \times_M \Lambda^{m-1} M)\,,\ \T\pi(\xi)=0\,,\ \xi|_{\partial D}=0\}\,.
\end{align*}

We want to find the sections which are ``critical'' for the action $\mathcal{A}$ under the constraint $\Phi$. As we have commented before, this problem is not well formulated. The constraint $\Phi$ involve velocities, and there are several non-equivalent ways to select which variations have to be taken \cite{Gra2003}. Inspired by the case of contact mechanics \cite{de_leon_constrained_2021}, we will describe two different non-equivalent approaches: the {\it non-holonomic} and the {\it vakonomic} variational principles.

\subsection{Herglotz principle for fields: non-holonomic version}

The approach presented in this section is inspired on \cite{binz_nonholonomic_2002,vankerschaver_geometric_2005, Gra2003}.
Let $D \subseteq M$ be an oriented manifold with compact closure and boundary $\partial D$. The vertical lift \cite{peretesis} is a morphism of vector bundles $
\mathcal{S}:\T^* M\otimes_{J^1(E \times_M \Lambda^{m-1}M)}V(\pi)\rightarrow V(\pi^1)$ over the identity of $J^1(E \times_M \Lambda^{m-1}M)$ such that, for any $j^1_x\phi\in J^1(E \times_M \Lambda^{m-1}M)$, $\beta\in \T^*M\otimes_{J^1(E \times_M \Lambda^{m-1}M)}V(\pi)$ and $f\in \Cinfty(J^1_{\phi(x)}(E \times_M \Lambda^{m-1}M))$, we have:
$$
\mathcal{S}_{j^1_x\phi}(\beta)(f)=\left.\frac{\dd}{\dd t}\right|_{t=0}f(j^1_x\phi+t\beta)\,.
$$
We have that $\T^*M\otimes_{J^1(E \times_M \Lambda^{m-1}M)}V(\pi)$ is the vector bundle associated to the affine bundle $\pi^1:J^1(E \times_M \Lambda^{m-1}M)\rightarrow E\times_M\Lambda^{m-1}M$ and, hence, using the same coordinates $(x^\mu, u^a,z^\nu, u^a_\mu,z^\nu_\mu)$, the local expression of the vertical lift is
$$
\mathcal{S}=\dd u^a\otimes \frac{\partial}{\partial x^\mu}\otimes\frac{\partial}{\partial u^a_\mu}+\dd z^\nu\otimes \frac{\partial}{\partial x^\mu}\otimes\frac{\partial}{\partial z^\nu_\mu}\,.
$$

The dependence on the velocities of the constraint is  implemented in the non-holonomic version as a force. This can be formalized in different ways. For instance, in \cite{binz_nonholonomic_2002} the authors use the vertical endomorphism. In our problem the constraint is given by the $m$-form $\Phi$ instead of a function, and we find that the vertical lift gives a more direct derivation of the equations. The vertical lift is a $(2,1)$-tensor, and we are interested in the contraction of both contravariant entrances with the form $\dd \Phi$:
$$
\varphi=i_{\mathcal{S}}\dd\Phi=\left(\frac{\partial \phi}{\partial u^a_\mu}\dd u^a+\frac{\partial \phi}{\partial z^\nu_\mu}\dd z^\nu\right)\otimes \dd^{m-1}x_\mu\,.
$$

\begin{definition}\label{dfn:noholonomicfield} A section
$\rho\in\Omega$ satisfies the \emph{non-holonomic Herglotz variational principle} if
\begin{equation}
\T_\rho\mathcal A(\xi)=\int_D (j^1\rho)^*(\lieD{\xi^1}\mathcal{L})=0\,.
\end{equation}
for all vector fields $\xi\in \T_\rho\Gamma_D$ such that
\begin{equation}  \label{eq:nonholonomic}
\varphi(\xi^1)=0\,,
\end{equation}
\end{definition}
where $\lieD{}$ denotes the Lie derivative and $\varphi(\xi^1)$ is the contraction of $\xi^1$ with the first entrance of $\varphi$. In local coordinates, it reads
$$
\varphi(\xi^1)=\left(\frac{\partial \phi}{\partial u^a_\mu}\xi^a+\frac{\partial \phi}{\partial z^\nu_\mu}\xi^\nu\right)\otimes \dd^{m-1}x_\mu\,.
$$

\begin{theorem}\label{thm:Herglotz_principle_fields}
    Let $\mathcal{L}: J^1(E\times_M\Lambda^{m-1}M) \to \Lambda^m M$ be a Lagrangian density  and let $\rho\in \Omega$. Then, $j^1\rho$ satisfies the Herglotz field equations
    \begin{align}\label{eq:herFilEq}
            D_\mu \left( \frac{\partial L}{\partial u^a_\mu} \right) - \frac{\partial L}{\partial u^a}  &= \frac{\partial L}{\partial u^a_\mu} \frac{\partial L}{\partial z^\mu}
    \end{align}
    if, and only if, $\rho$ satisfies the non-holonomic Herglotz variational principle (Definition \ref{dfn:noholonomicfield}).
\end{theorem}
\begin{proof}

\begin{align*}
 \int_D (j^1\rho)^*(\lieD{\xi^1}\mathcal{L})&=\int_D (j^1\rho)^*\left[\xi^a\frac{\partial L}{\partial u^a} +\left(\frac{\partial\xi^a}{\partial x^\mu}+u^b_\mu\frac{\partial\xi^a}{\partial u^b}+z^\tau_\mu\frac{\partial\xi^a}{\partial z^\tau}\right)\frac{\partial L}{\partial u^a_\mu}\right.
 \\&\quad\left.+\xi^\nu\frac{\partial L}{\partial z^\nu} +\left(\frac{\partial\xi^\nu}{\partial x^\mu}+u^b_\mu\frac{\partial\xi^\nu}{\partial u^b}+z^\tau_\mu\frac{\partial\xi^\nu}{\partial z^\tau}\right)\frac{\partial L}{\partial z^\nu_\mu}\right]\dd^mx
  \\
  &=\int_D (j^1\rho)^*\xi^a\frac{\partial L}{\partial u^a} +\frac{\partial\xi^a\circ\rho}{\partial x^\mu}(j^1\rho)^*\frac{\partial L}{\partial u^a_\mu}
+(j^1\rho)^*\xi^\nu\frac{\partial L}{\partial z^\nu} +\frac{\partial\xi^\nu\circ\rho}{\partial x^\mu}(j^1\rho)^*\frac{\partial L}{\partial z^\nu_\mu}\dd^mx
  \\
  &=\int_D(j^1\rho)^*\left[ \xi^a\left(\frac{\partial L}{\partial u^a} -D_\mu\frac{\partial L}{\partial u^a_\mu}\right)+\xi^\nu\frac{\partial L}{\partial z^\nu} \right]\dd^mx+\int_{\partial D}(j^1\rho)^* \xi^a\frac{\partial L}{\partial u^a_\mu}\dd^{m-1}x_\mu
  \\
  &=\int_D(j^1\rho)^*\left[ \xi^a\left(\frac{\partial L}{\partial u^a} -D_\mu\frac{\partial L}{\partial u^a_\mu}\right)+\xi^\nu\frac{\partial L}{\partial z^\nu} \right]\dd^mx\,.
\end{align*}
If it vanishes for all $\xi$ satisfying equation \eqref{eq:nonholonomic}, there exist functions $\lambda_\alpha\in \Cinfty(J^1(E\times_M\Lambda^{m-1}M))$ such that
\begin{align*}
    \frac{\partial L}{\partial u^a} -D_\mu \left( \frac{\partial L}{\partial u^a_\mu} \right) &= \lambda_\alpha\frac{\partial \phi}{\partial u^a_\alpha}  \,,
 \\
  \frac{\partial L}{\partial z^\nu}  &= \lambda_\alpha \frac{\partial \phi}{\partial z^\nu_\alpha} \,.
\end{align*}
Combining both equations and using the expression $\phi=z^\mu_\mu-L$, we see that $\lambda_\mu=\dfrac{\partial L}{\partial z^\mu}$ and 
    \begin{align}
            D_\mu \left( \frac{\partial L}{\partial u^a_\mu} \right) - \frac{\partial L}{\partial u^a}  &= \frac{\partial L}{\partial u_\mu^a}  \frac{\partial L}{\partial z^\mu}\,.
    \end{align}
\end{proof}
The Herglotz field equations \eqref{eq:herFilEq} are also called $k$-contact Euler--Lagrange equations \cite{Gaset2021}.

\subsection{Herglotz principle for fields: vakonomic version}\label{section:LagMult}

The approach presented in this section is inspired by the vakonomic version of Herglotz principle \cite{de_leon_constrained_2021}, presented in Section \ref{subs:mechanicsLagMult}. 

In the vakonomic approach we only consider variations that transform sections that satisfy the constraints into sections that also satisfy the constraints. In other words, the lift of the variations to the first jet must be tangent to the submanifold defined by the constraints. Thus, we have the following variational principle. Let $D \subseteq M$ be an oriented manifold homeomorphic to a ball and with boundary $\partial D$.

\begin{definition}\label{dfn:LagMultfield} A section
$\rho\in\Omega$ satisfies the \emph{vakonomic Herglotz variational principle} if
\begin{equation}
\T_\rho\mathcal A(\xi)=\int_D (j^1\rho)^*(\lieD{\xi^1}\mathcal{L})=0
\end{equation}
for every vector field $\xi\in \T_\rho\Gamma_D$ such that $\lieD{\xi^1}\Phi=0$.
\end{definition}

This kind of constrained field theories has been studied, for instance, in \cite{campos_vakonomic_2010}. By Theorem \ref{thm:Lagrange_multipliers}, we can rewrite this as a problem without constraints using Lagrange multipliers. We need to consider the Lagrangian
$$
\L_\lambda=\L+\lambda\Phi=\left(L+\lambda(z^\mu_\mu-L)\right)\dd^mx = L_\lambda\dd^m x\,,
$$
where $\lambda\in \Cinfty(M)$ is a function to be determined called the {\it Lagrange multiplier}. Then, the action associated to $\mathcal{L}_\lambda$ is
\begin{align*}
\mathcal{A}_\lambda:\Omega&\longrightarrow \mathbb{R}
\\
    \rho & \longmapsto\int_D (j^1\rho)^*\mathcal{L}_\lambda\,.
\end{align*}

\begin{corollary} A section
$\rho\in\Omega$ satisfies the vakonomic Herglotz variational principle if, and only if,
\begin{equation}
\T_\rho\mathcal A_\lambda(\xi)=\int_D (j^1\rho)^*(\lieD{\xi^1}\mathcal{L}_\lambda) = 0
\end{equation}
for every vector field $\xi\in \T_\rho\Gamma_D$.
\end{corollary}

The corresponding equations are given by the following theorem.

\begin{theorem}\label{thm:Herglotz_principle_fields-LagMult}
    Let $\mathcal{L}: J^1(E\times_M\Lambda^{m-1}M) \to \Lambda^m M$ be a Lagrangian density  and let $\rho\in \Omega$. Then, $j^1\rho$ satisfies the Herglotz field equations:
   \begin{align}  \label{eq:LagMulHer}
            D_\mu \left( \frac{\partial L}{\partial u^a_\mu} \right) - \frac{\partial L}{\partial u^a}  &= \frac{\partial L}{\partial u_\mu^a}  \frac{\partial L}{\partial z^\mu}\,,
    \end{align}
   and the condition
      \begin{align}\label{condcerrada}
            D_\nu\frac{\partial L}{\partial z^\mu}  &=
            D_\mu\frac{\partial L}{\partial z^\nu}\,,
    \end{align}
    if, and only if, $\rho$ satisfies the vakonomic Herglotz variational principle.
\end{theorem}
\begin{proof}

The problem is the usual non-constrained Hamilton variational problem for the Lagrangian $L_\lambda$. Considering variations with respect to $\delta u^a$ and $\delta z^\nu$ we obtain the set of equations
\begin{align}\label{eq:LagMult1}
    \frac{\partial L_\lambda}{\partial u^a} -D_\mu \left( \frac{\partial L_\lambda}{\partial u^a_\mu} \right) &=0  \,,
 \\\label{eq:LagMult2}
    \frac{\partial L_\lambda}{\partial z^\nu} -D_\mu \left( \frac{\partial L_\lambda}{\partial z^\nu_\mu} \right) &=0  \,.
\end{align}
These equations are just the Euler--Lagrange equations when considering $u^a$ and $z^\nu$ as dynamical variables. Expanding equation \eqref{eq:LagMult2}, we have
$$
(1-\lambda)\frac{\partial L}{\partial z^\nu} -\frac{\partial \lambda}{\partial x^\mu}=0\,,
$$
and combining it with equation \eqref{eq:LagMult1}, we find that
$$
0=(1-\lambda)\frac{\partial L}{\partial u^a} -D_\mu \left((1-\lambda) \frac{\partial L}{\partial u^a_\mu}\right)=(1-\lambda)\frac{\partial L}{\partial u^a} -(1-\lambda)D_\mu \left( \frac{\partial L}{\partial u^a_\mu}\right) +(1-\lambda)\frac{\partial L}{\partial z^\nu} \frac{\partial L}{\partial u^a_\mu}\,.
$$
If $\lambda\neq1$, we can divide by $1-\lambda$ and obtain equation \eqref{eq:LagMulHer}. However, in this case there are hidden conditions in equation \eqref{eq:LagMult2}. Taking $g  = \log(\abs{1-\lambda})$, equation \eqref{eq:LagMult2} implies
\begin{equation*}
        \dd g  = \pm \frac{\partial L}{\partial z^\nu}\dd x^\nu\,.
\end{equation*}
This has solution if and only if the right hand side is closed, namely if
\begin{equation}
        D_\nu\frac{\partial L}{\partial z^\mu}  =
        D_\mu\frac{\partial L}{\partial z^\nu}\,.\label{eq:z_var}
\end{equation}
If this condition is fulfilled, since $D$ is homeomorphic to a ball,
\begin{equation}
    \frac{\partial L}{\partial z^\nu}\dd x^\nu = \dd h\,,
\end{equation}
and so we pick $g = h$.
\end{proof}

\subsection{Relations between both approaches}
The main difference between the non-holonomic and the vakonomic approaches is the unexpected condition \eqref{condcerrada}. It motivates the following definition.
\begin{definition}\label{dfn:ClosedLag}
A Lagrangian has \emph{closed action dependence} if
\begin{equation}\label{eq:closed_action}
    D_\mu\frac{\partial L}{\partial z^\nu}= D_\nu\frac{\partial L}{\partial z^\mu}
\end{equation}
for any pair $1\leq\mu,\nu\leq m$.
\end{definition}

This condition has two interesting interpretations: a variational one and a geometric one. The Lagrangian has closed action dependence if, and only if, the action of $\rho=(\sigma,\zeta) \in \Omega$ only depends on $\sigma$. Equation \eqref{eq:z_var} is obtained by taking variations of the constrained action in  the ${\zeta}$ direction. Indeed, a Lagrangian has closed action dependence if, and only if, for any section $\sigma: M \to E$, does not exist a family of sections $\zeta_s : M \to \Lambda^{m-1} (M)$, $s \in \mathbb{R}$, such that
$$ \left.\dfrac{\dd \zeta_s}{\dd s}\right|_{\partial D}=0$$
and satisfying the conditions
$$ \dd \zeta_s = \mathcal{L}(j^1 \sigma, \zeta_s)\quad\text{and}\quad\dfrac{\partial \mathcal{A(\sigma, \zeta_s)}}{\partial s}\bigg|_{s=0} \neq 0\,.$$
The reason is because, if $\frac{\partial L}{\partial z^\nu}$ induces a closed form, by Stokes' theorem the action only depends on the border, where the variation vanishes. This can be seen explicitly in the example presented in Section \ref{ex:contraejemplo}.

The geometric interpretation can be obtained as follows. Let $L: J^1 (E \times_M \Lambda^{m-1}M) \to \mathbb{R}$ be a Lagrangian function. Define the $M$-semibasic one-form $\theta_{L} \in \Omega^1(J^1( E \times_M \Lambda^{m-1}M)) $ as 
\begin{equation}
    \theta_{L} = \frac{\partial L}{\partial z^\mu} \dd x^\mu,
\end{equation}
which is independent on the coordinates used to define it. The closed action dependence condition is equivalent to
\begin{equation}
    \dd \theta_{L} = 0.
\end{equation}
The form $\theta_L$ is (minus) the dissipation form introduced in \cite{de_leon_multicontact_2022}.

For Lagrangians with closed action dependence, both versions of the variational principle given in Definitions \ref{dfn:noholonomicfield} and \ref{dfn:LagMultfield} are equivalent. Moreover, they coincide with the version proposed in \cite{Lazo2018} and the equations are the same as the ones derived from the $k$-contact \cite{Gaset2021} and multicontact \cite{de_leon_multicontact_2022} formalisms.

When the Lagrangian has not closed action dependence, both principles may be different. In Section \ref{ex:contraejemplo}, we provide an example where there are sections which are solutions of one variational principle but not the other. In this case, only the non-holonomic approach provides, in general, the same equations as the $k$-contact and multicontact formalisms.

\section{Higher-order Lagrangian densities}\label{sec:4}

Most of the relevant field theories are modelled by first-order Lagrangians with one notable exception, General Relativity, which is usually described with a second-order Lagrangian. Contact gravity is specially interesting as an example of modified gravity which may explain certain observations about the expansion of the universe \cite{paiva_generalized_2022}. The Herglotz field equations for the Hilbert Einstein Lagrangian with a linear term in the action have been derived in \cite{paiva_generalized_2022} and \cite{gaset_EH_2022} with slightly different variational methods. The method used in  \cite{gaset_EH_2022} is, essentially, the vakonomic method presented in Section \ref{section:LagMult}, showing how it can be expanded to higher-order Lagrangians. Hence, in this section we apply the vakonomic principle to higher-order Lagrangian densities. 

Consider the $r$-th jet bundle $J^r(E \times_M \Lambda^{m-1}M)$ of a fiber bundle $E \to M$. Local coordinates of $J^r (E \times_M \Lambda^{m-1}M)$ will be denoted as $(x^\mu, u^a_I)$, where $I = (I_1, \ldots, I_m)$ is a multi-index such that $0 \leq \abs{I} = I_1 + \ldots + I_n \leq r$. Given a local section $\sigma:M \to E$, we denote by $j^r \sigma: M \to J^r E$ its $r$-th prolongation.

Given a coordinate system, the total derivative $D_\mu:\Cinfty(J^k(E \times_M \Lambda^{m-1}M))\rightarrow \Cinfty(J^{k+1}(E \times_M \Lambda^{m-1}M))$, for $\mu=1,\dots,m$, is a derivation given by
\begin{align*}
D_\mu f=\frac{\partial f}{\partial x^\mu}+\sum_{|J|=0}^k\left(u^a_{J+1_\mu}\frac{\partial f}{\partial u^a_J}+z^\nu_{J+1_\mu}\frac{\partial f}{\partial z^\nu_J}\right)\,,
\end{align*}
where $f\in \Cinfty(J^k(E \times_M \Lambda^{m-1}M))$.

The Lagrangian density $\mathcal{L}: J^r (E \times_M \Lambda^{m-1}M) \to \Lambda^m M$ is a fiber bundle morphism over $M$. Locally, $\mathcal{L}=L\dd^mx$. The Herglotz operator \cite{DeLeo2021} can be extended to fields.

\begin{definition}
Given a Lagrangian $\mathcal{L}$ and an index $1\leq\mu\leq m$, the \emph{Herglotz operator} for fields is the linear operator
\begin{align*}
    D^{\mathcal{L}}_\mu: \Cinfty(J^r( E \otimes_M \Lambda^{m-1} M))&\longrightarrow \Cinfty (J^{r+1} (E \otimes_M \Lambda^{m-1} M))
    \\
    F&\longmapsto D^{\mathcal{L}}_\mu (F)=D_\mu F - F \frac{\partial L}{\partial z^\mu}\,.
\end{align*}
\end{definition}

In general, these operators are not derivations and, since
\begin{equation*}
    \left(D^{\mathcal{L}}_\mu D^{\mathcal{L}}_\nu-D^{\mathcal{L}}_\nu D^{\mathcal{L}}_\mu\right) F=\left(D_\nu \frac{\partial L}{\partial z^\mu}-D_\mu\frac{\partial L}{\partial z^\nu}\right)F\,,
\end{equation*}
they do not commute.
\begin{lemma}
The Herglotz operators commute if, and only if, the Lagrangian has closed action dependence.
\end{lemma}
For Lagrangians with closed action dependence, we can denote the successive applications of the Herglotz operator with multi-index notation as

\begin{equation*}
    D^{\mathcal{L}}_I =\prod_{\mu=1}^m \left(D^{\mathcal{L}}_\mu\right)^{I_\mu}\,.
\end{equation*}
The constraint is implemented as in the first-order case, that is
\begin{equation}
    \Phi=(\tau^r_1)^*\overline{DS}-\mathcal{L}=0\,,
\end{equation}
where $(\tau^r_1):J^r(E\times_M\Lambda^{m-1}M)\rightarrow J^1\Lambda^{m-1}M$ is the projection. In local coordinates, $\Phi=\phi\dd^mx$, with $\phi=z^\mu_\mu-L$. 
Let $D \subseteq M$ be an oriented manifold homeomorphic to a ball and with boundary $\partial D$. The set of sections on $D$ which satisfy the constraint is denoted by
\begin{equation*}
 \Omega= \set{\rho \in \Gamma_D(E \times_M \Lambda^{m-1} M)\, \text{ such that }\, (j^r\rho)^*\Phi=0} .
\end{equation*}

In the following definition we introduce the higher-order version of the vakonomic variational principle presented in Definition \ref{dfn:LagMultfield}.

\begin{definition}\label{dfn:LagMultfieldHO}
    A section $\rho\in\Omega$ satisfies the \emph{higher-order vakonomic Herglotz variational principle} if
    \begin{equation}
        \T_\rho\mathcal A(\xi)=\int_D (j^r\rho)^*(\lieD{\xi^r}\mathcal{L})=0\,,
    \end{equation}
    for every vector field $\xi\in \T_\rho\Gamma_D$ such that $\lieD{\xi^r}\Phi=0$.
\end{definition}

As before, we have an equivalent version of this variational principle based on Lagrange multipliers \cite{campos_vakonomic_2010}. Consider the modified Lagrangian
$$
\L_\lambda=\L+\lambda\Phi=\left(L+\lambda(z^\mu_\mu-L)\right)\dd^mx = L_\lambda\dd^m x\,.
$$
Then, the action associated to $\mathcal{L}_\lambda$ is
\begin{align*}
\mathcal{A}_\lambda:\Omega & \longrightarrow \mathbb{R}
\\
    \rho & \longmapsto\int_D (j^1\rho)^*\mathcal{L}_\lambda\,.
\end{align*}

\begin{corollary}
    A section $\rho\in\Omega$ satisfies the higher-order vakonomic Herglotz variational principle if, and only if,
    \begin{equation}
        \int_D (j^r\rho)^*(\lieD{\xi^r}\mathcal{L}_\lambda)=0\,,
    \end{equation}
    for every vector field $\xi\in \T_\rho\Gamma_D$.
\end{corollary}

\begin{theorem}\label{thm:Herglotz_principle_fields-higher-order} Let $\mathcal{L}: J^r( E \otimes_M \Lambda^{m-1} M) \to \Lambda^m M$ be a Lagrangian density  and let $\rho\in \Omega$. Then, $j^r\rho$ satisfies the higher-order Herglotz field equations
    \begin{align*}  
           \sum_I {(-1)}^{\abs{I}} D_I^\mathcal{L} \left(
        \frac{\partial L}{\partial u^a_I}  \right)  &= 0
    \end{align*}
    and the condition
        \begin{align*}  
            D_\nu\frac{\partial L}{\partial z^\mu}  =
            D_\mu\frac{\partial L}{\partial z^\nu} \,,
    \end{align*}
    if, and only, if $\rho$ satisfies the higher-order vakonomic Herglotz variational principle.
\end{theorem}

\begin{proof}
    We proceed in a similar way to the first-order case.
The Euler--Lagrange equations of $\mathcal{L}_\lambda$ are given by
    \begin{align}\label{eq:euler_langrange_contact_fields_q_ho}
        \sum_I {(-1)}^{\abs{I}} D_I \left(\lambda  
        \frac{\partial L_\lambda}{\partial u^a_I}  \right)  &= 0\,, \\\label{eq:euler_langrange_contact_fields_z_ho}
    \lambda 
        \frac{\partial  L_\lambda}{\partial z^\nu}  -D_\mu\left(\lambda 
        \frac{\partial  L_\lambda}{\partial z^\nu_\mu}  \right) &= 0\,.
     \end{align}
Since $L_\lambda$ does only depend of $\zeta$ and its first derivatives, higher-order terms in equation \eqref{eq:euler_langrange_contact_fields_z_ho} vanish. Taking into account the definition of $L_\lambda$, we have
\begin{equation}\label{eq:lambda_fields_ho}
(1-\lambda)\frac{\partial L}{\partial z^\nu} -\frac{\partial \lambda}{\partial x^\mu}=0\,.
\end{equation}

Repeating the argument used in the first-order case, this has solution $\lambda(x^\mu)$ if and only if
\begin{equation*}
        D_\nu\frac{\partial L}{\partial z^\mu}  =
        D_\mu\frac{\partial L}{\partial z^\nu}\,.
\end{equation*}
That is, there only exist solutions where $\mathcal{L}$ has closed action dependence. Hence,
by equation~\eqref{eq:lambda_fields_ho}, we see that, for any function $F$,
\begin{equation*}
    D_\mu\left( (1-\lambda) F\right)=-D_\mu\left( \lambda F\right) =- \lambda D^{\mathcal{L}}_\mu F\,.
\end{equation*}
Substituting the above expression in ~\eqref{eq:euler_langrange_contact_fields_q_ho}, we obtain the higher-order Herglotz field equations.
\end{proof}

These equations are compatible with the ones derived in \cite{gaset_EH_2022} for the Hilbert--Einstein Lagrangian.

\section{Examples}\label{sec:5}

\subsection{Vibrating string with damping}

In this example we are going to study how we can derive the equation of a vibrating string with damping from a Herglotz principle. It is well known that a vibrating string can be described using the Lagrangian formalism. Consider the coordinates $(t,x)$ for the time and the space. Denote by $u$ the separation of a point in the string from its equilibrium point, and hence $u_t$ and $u_x$ will denote the derivative of $u$ with respect to the two independent variables. The Lagrangian function for this system is
\begin{equation}\label{eq:Lagrangian-vibrating-string}
    L_0(u,u_t,u_x) = \frac{1}{2}\rho u_t^2 - \frac{1}{2}\tau u_x^2\,,
\end{equation}
where $\rho$ is the linear mass density of the string and $\tau$ is the tension of the string. We will assume that these quantities are constant. The Euler--Lagrange equation for this Lagrangian function is
$$ u_{tt} = c^2 u_{xx}\,, $$
where $c^2 = \dfrac{\tau}{\rho}$.

In order to model a vibrating string with linear damping, we can modify the Lagrangian function \eqref{eq:Lagrangian-vibrating-string} so that it becomes a $k$-contact Lagrangian \cite{Gaset2021}.

The new Lagrangian function $L$ is defined in the phase bundle $\oplus^2\T Q\times \RR^2$, equipped with adapted coordinates $(u; u_t, u_x; z^t, z^x)$, and is given by
$$ L(u, u_t, u_x, z^t, z^x) = L_0 - \gamma z^t = \frac{1}{2}\rho u_t^2 - \frac{1}{2}\tau u_x^2 - \gamma z^t\,, $$
where $\gamma\in\RR$ is a constant accounting for the damping.

The Herglotz equation \eqref{eq:Herglotz-equations} for this Lagrangian $L$ reads
$$ u_{tt} = c^2 u_{xx} - \gamma u_t\,, $$
which is the equation of a vibrating string with damping. The additional equation \eqref{condcerrada},
$$
    D_\nu\frac{\partial L}{\partial z^\mu} = D_\mu\frac{\partial L}{\partial z^\nu}\Leftrightarrow \begin{dcases}
        -\partial_t\gamma = 0\,,\\
        -\partial_x\gamma = 0\,,
    \end{dcases}
$$
is trivially satisfied since $\gamma$ is constant, and hence the equations obtained are exactly the same as in the $k$-contact Lagrangian formalism introduced in \cite{Gaset2021}. The next example presents a case in which both approaches are not fully equivalent.

\subsection{The non-holonomic and the vakonomic principles are not equivalent}\label{ex:contraejemplo}

Consider the Lagrangian 
$$L(t,x,u,u_t,u_x,z^t,z^x)=\frac12(u_t^2+
u_x^2)-u\gamma_xz^x\,,
$$
where $\gamma_x\neq 0$ is a constant. The Lagrangian function $L$ is regular in the sense of \cite{de_leon_multicontact_2022,Gaset2021}. This Lagrangian has not closed dependence action. The corresponding Hergltoz field equations are
\begin{align}\nonumber
    \gamma_xz^x+u_{xx}+u_{tt}+u\gamma_x u_x&=0\,,
    \\\nonumber
    z^t_t+z^x_x&=L\,.
\end{align}
A solution of these equations is the section $u(t,x)=t$, $z^x(t,x)=0$ and $z^t(t,x)=\frac{t}{2}$. Nevertheless, for this section, we have

$$ D_t\frac{\partial L}{\partial z^x} =D_x\frac{\partial L}{\partial z^t} \Rightarrow \gamma_x u_t=0\Rightarrow \gamma_x =0\,, $$
which is not satisfied as long as $\gamma_x\neq0$. Therefore, this section is a solution of the non-holonomic variational principle, but it is not a solution of the vakonomic variational principle. Therefore, both principles are not equivalent.




\subsection{The Korteweg--De Vries Lagrangian}
The Korteweg--De Vries (KdV) equation is used to model waves on shallow water \cite{Kor1895}. This equation can be derived as the Euler-Lagrange equation of a second order Lagrangian. We will use the higher-order vakonomic Herglotz variational principle introduced in Definition \ref{dfn:LagMultfieldHO} to derive the equations of motion of a contact analogue of the KdV Lagrangian.

KdV equation involves a scalar field over time and one dimension of space. Therefore, we consider a $2$-dimensional base manifold $M$, with coordinates $(t,x)$. Then, in the second order jet $J^2(\mathbb{R}\otimes\Lambda^2M)$ we consider the coordinates 
$$(t,x,u,u_t, u_x, u_{xx}, u_{xy}, u_{yy}, z^t, z^x,z^t_t,z^t_x,z^x_t,z^x_x,z^t_{tt},z^t_{tx},z^t_{xx},z^x_{tt},z^x_{tx},z^x_{xx} )\,.$$
The standard KdV Lagrangian is
\begin{equation}
    L_0 = \frac{1}{2}u_x u_t + u_x^3 - \frac{1}{2}u_{xx}^2\,.
\end{equation}
The Euler--Lagrange equation one obtains from this Lagrangian is
\begin{equation}
    \partial_t \partial_x u + 6 \partial_x u \partial_x^2 u + \partial_x^4 u = 0\,.
\end{equation}
Let us now consider the KdV Lagrangian with a linear action coupling
\begin{equation}
    L = L_0 - \gamma_\mu z^\mu = \dfrac{1}{2}u_x u_t + u_x^3 - \dfrac{1}{2}u_{xx}^2
    - \gamma_\mu z^\mu\,.,
\end{equation}
which has closed action dependence provided that $\gamma_\mu$ are the components of a closed form.
This is a second order Lagrangian, so we need to use the Herglotz field equations derived in
\cref{thm:Herglotz_principle_fields-higher-order}. The Herglotz field equation reads
\begin{equation}
  \partial_t \partial_x u + \dfrac{1}{2}(\gamma_x \partial_tu + \gamma_t \partial_xu) + 6
  \partial_x u \partial_x^2 u + 3\gamma_x(\partial_x u)^2 + \partial_x^4 u +
  (2\gamma_x + \partial_x \gamma_x)\partial_x^3 u + \gamma_x^2 \partial_x^2 u = 0\,,
\end{equation}
along with the constraint
$$
z^t_t+z^x_x=L\,.
$$
One sees that there are additional terms which are linear in the \( \gamma_\mu \), which also appear in the first-order theory, as well as quadratic terms in \( \gamma_\mu \) and involving their derivatives, which are characteristic of a second-order theory.

\section{Conclusions and outlook}

In this paper we have developed a generalization of the Herglotz variational principle \cite{dLeo2021,Her1930} for first-order and higher-order field theories. In order to do this, we have developed two non-equivalent approaches: the non-holonomic and the vakonomic versions. We have seen that the non-holonomic approach is equivalent to the $k$-contact \cite{Gaset2021,Riv2022} and multicontact \cite{de_leon_multicontact_2022} geometric formulations of dissipative field theories. On the other hand, using the vakonomic principle, some new conditions arise. This fact motivates the introduction of the so-called Lagrangians with closed action dependence, for which both approaches are equivalent.

The differences between the non-holonomic and the vakonomic principles have been exemplified with an academic example which has a solution to its $k$-contact Euler--Lagrange equations that is not a solution to the Herglotz field equations arising from the vakonomic variational principle. This is because the Lagrangian considered has not closed action dependence.

We have also studied a first-order field theory, the damped vibrating string, for which the $k$-contact formalism and the Herglotz variational principle are fully equivalent. The last example consisted in modifying the Korteweg--De Vries Lagrangian by adding a standard dissipative term.

In \cite{Lazo2018}, a variational principle for Lagrangians with closed action dependence is derived using an implicit argument. The extension of this approach to the general case will require a deeper analysis of equation \eqref{eq:zmu_def0}. This might be clarifying to understand the condition of closed action dependence \eqref{eq:closed_action}.

There are still many open problems in the geometrization of action-dependent field theories. In first place, it would be interesting to establish the relations among the different geometric frameworks ($k$-contact, $k$-cocontact and multicontact) and the variational principles presented in this work and previous one. Another relevant problem is the case of field theories described by singular Lagrangians. 

There are some singular Lagrangians which are not compatible with the current geometric structures, not even a weakened version of them \cite{DeLeo2022}. Nevertheless, we can derive their corresponding field equations via variational principles. We expect this work will help in the understanding of the underlying geometric structures of these singular Lagrangians.

\subsection*{Acknowledgments}

The authors acknowledge fruitful discussions and comments from our colleague Miguel-C. Mu\~noz-Lecanda.

J. Gaset and X. Rivas acknowledge partial financial support from the Ministerio de Ciencia, Innovaci\'on y Universidades (Spain), projects PGC2018-098265-B-C33 and D2021-125515NB-21.

M. Lainz acknowledges partial financial support of the Spanish Ministry of Science and Innovation (MCIN/AEI/ 10.13039/501100011033), under grants PID2019-106715GB-C2 and ``Severo Ochoa Programme for Centres of Excellence in R\&D'' (CEX2019-000904-S).

X. Rivas acknowledges partial financial support from Novee Idee 2B-POB II project PSP: 501-D111-20-2004310 funded by the ``Inicjatywa Doskonałości - Uczelnia Badawcza'' (IDUB) program.

\bibliographystyle{abbrv}
\bibliography{Fields}

\addcontentsline{toc}{section}{References}

\end{document}